\documentclass[a4paper,UKenglish]{lipics-v2019}
\newcommand{\h}{\hspace*{0.2in}}

\usepackage{thm-restate}
\usepackage{makecell}
\nolinenumbers

\title{Byzantine Lattice Agreement in Synchronous Systems} 

\titlerunning{Byzantine Lattice Agreement in Synchronous Systems}

\author{Xiong Zheng}{University of Texas at Austin, Austin, TX 78712, USA}{zhengxiongtym@utexas.edu}{}{}

\author{Vijay K. Garg}{University of Texas at Austin, Austin, TX 78712, USA}{garg@ece.utexas.edu}{}{}

\authorrunning{X.\,Zheng, and V.\,K. Garg} 

\Copyright{Xiong Zheng, and Vijay K. Garg}



\category{}

\relatedversion{}

\supplement{}

\acknowledgements{}

\begin{document}

\maketitle

\begin{abstract}
In this paper, we study the Byzantine lattice agreement problem in synchronous distributed message passing systems. The lattice agreement problem \cite{attiya1995atomic} in crash failure model has been studied both in synchronous and asynchronous systems 
\cite{attiya1995atomic, faleiro2012generalized, zheng2018lattice, zheng2018linearizable}, which leads to the current best upper bound of $O(\log f)$ rounds both in synchronous and asynchronous systems. Its applications in building linearizable replicated state machines has also been further explored recently in \cite{faleiro2012generalized, skrzypczak2019linearizable, zheng2018linearizable}. However, very few algorithmic results are known for the lattice agreement problem in Byzantine failure model. The paper by Nowak et al~\cite{nowak2019byzantine} first gives an algorithm for a variant of the lattice agreement problem on cycle-free lattices that tolerates up to $f < n/(h(X) + 1)$ Byzantine faults, where $n$ is the number of processes and $h(X)$ is the height of the input lattice $X$. The recent preprint by Di Luna et al~\cite{di2019byzantine} studies this problem in asynchronous systems and slightly modifies the validity condition of the original lattice agreement problem in order to accommodate extra values sent from Byzantine processes. They present a $O(f)$ rounds algorithm by using the reliable broadcast primitive as a first step and following the similar algorithmic framework as in \cite{faleiro2012generalized, zheng2018lattice}.  

In this paper, we propose three algorithms for the Byzantine lattice agreement problem in synchronous systems. The first algorithm runs in $\min \{3h(X) + 6,6\sqrt{f} + 6\})$ rounds and takes $O(n^2 \min\{h(X), \sqrt{f}\})$ messages, where $h(X)$ is the height of the input lattice $X$, $n$ is the total number of processes in the system and $f$ is the maximum number of Byzantine processes such that $n \geq 3f + 1$. The second algorithm takes $3\log n + 3$ rounds and $O(n^2 \log n)$ messages. The third algorithm takes $4 \log f + 3$ rounds and $O(n^2 \log f)$ messages. All algorithms can tolerate $f < \frac{n}{3}$ Byzantine failures. 
In our algorithms, we apply a slightly modified version of the Gradecast algorithm given by Feldman et al \cite{feldman1988optimal} as a building block. If we use the Gradecast algorithm for authenticated setting given by Katz et al \cite{katz2006expected}, we obtain algorithms for the Byzantine lattice agreement problem in authenticated settings and tolerate $f < \frac{n}{2}$ failures. 
\end{abstract}

\section{Introduction}
The lattice agreement problem, introduced by Attiya et al~\cite{attiya1995atomic}, is an important decision problem in shared-memory and message passing systems. In this problem, processes start with input values from a lattice and need to decide values which are comparable to each other. Specifically, suppose each process $i$ has input $x_i$ from a lattice ($X$, $\leq$, $\sqcup$) with the partial order $\leq$ and the join operation $\sqcup$, it has to output a value $y_i$ also in $X$ such that the following properties are satisfied. 

1) {\bf Downward-Vaility}: $x_i \leq y_i$ for each correct process $i$. 

2) {\bf Upward-Validity}: $y_i \leq \sqcup \{x_i ~ | ~ i \in [n]\}$. 

3) {\bf Comparability}: for any two correct processes $i$ and $j$, either $y_i \leq y_j$ or $y_j \leq y_i$.\\ 

In shared-memory systems, algorithms for the lattice agreement problem can be directly applied to solve the atomic snapshot problem \cite{afek1993atomic, attiya1995atomic}. This was the initial motivation for studying this problem. The application of lattice agreement in message passing systems has been explored only recently. Failero et al \cite{faleiro2012generalized} were the first to apply lattice agreement for building a special class of linearizable replicated state machines, which can support query operation and update operation, but not mixed query and update operation. Traditionally, consensus based protocols are applied to build linearizable replicated state machines. However, consensus based protocols do not provide termination guarantee in the presence of failures in the system, since the consensus problem cannot be solved with even one failure in an asynchronous system \cite{fischer1985impossibility}. The lattice agreement instead has been shown to be a weaker decision problem than consensus. It can be solved in an asynchronous system when a majority of processes is correct. Thus, linearizable replicated state machines built based on lattice agreement protocols have the advantage of termination even with failures. Another application of lattice agreement in distributed systems is to build atomic snapshot objects. Efficient implementation of atomic snapshot objects in crash-prone asynchronous message passing systems is important because they can make design of algorithms in such systems easier. The paper \cite{attiya1995atomic} presents a general technique for applying algorithms for the lattice agreement problem to solve the atomic snapshot problem. By using the same technique, algorithms for lattice agreement problem in distributed systems can be directly applied to implement atomic snapshot objects in crash-prone message passing systems. Essentially, an atomic snapshot object needs to provide linearizabilty for all processes, which basically decides on some total ordering of operations. In the lattice agreement problem, processes need to output values which lie in a chain of the input lattice, which is also a total ordering. 

It has been shown that the lattice agreement problem is a weaker decision problem than consensus in the crash-failure model. In synchronous systems, consensus cannot be solved in fewer than $f + 1$ rounds~\cite{dolev1983authenticated},  but lattice agreement can be solved in $\log f$ rounds~\cite{zheng2018lattice}. In asynchronous systems, the consensus problem cannot be solved even with one failure~\cite{fischer1985impossibility}, whereas the lattice agreement problem can be solved in asynchronous systems when a majority of processes is correct~\cite{faleiro2012generalized}. In asynchronous systems with Byzantine failures, the recent preprint~\cite{di2019byzantine} has shown that the lattice agreement problem, with a slightly modified Upward-Validity definition, can be solved in $O(f)$ rounds, whereas Byzantine consensus cannot be solved even with one failure \cite{lamport1982byzantine}. For synchronous systems, we show in this work that the Byzantine lattice agreement problem can be solved in $O(\sqrt{f})$ rounds and $O(\log n)$ rounds. Byzantine consensus, however, cannot be solved in fewer than $f + 1$ rounds \cite{dolev1983authenticated}, where $f$ is the maximum number of Byzantine processes in the system such that $n \geq 3f + 1$. 


Our main contribution in this paper is summarized as below.
\begin{restatable}{theorem}{mainTheorem} \label{main_theorem}
There is a $\min\{3h(X) + 6,6\sqrt{f} + 6\})$ rounds early stopping algorithm for the Byzantine lattice agreement problem in synchronous systems, which can tolerate $f < \frac{n}{3}$ Byzantine failures. The term $h(X)$ is the height of the input lattice $X$. The algorithm takes $O(n^2 \min\{h(X), \sqrt{f}\})$ messages.
\end{restatable}

\begin{restatable}{theorem}{logTheorem} \label{log_theorem}
There is a $3\log n + 3$ rounds algorithm for the Byzantine lattice agreement problem in synchronous systems which can tolerate $f < \frac{n}{3}$ Byzantine failures, where $n$ is the number of processes in the system. The algorithm takes $O(n^2 \log n)$ messages. 
\end{restatable}

\begin{restatable}{theorem}{logfTheorem} \label{logf_theorem}
There is a $4\log f + 3$ rounds algorithm for the Byzantine lattice agreement problem in synchronous systems which can tolerate $f < \frac{n}{3}$ Byzantine failures, where $n$ is the number of processes in the system. The algorithm takes $O(n^2 \log n)$ messages. 
\end{restatable}


\subsection{Related Work}
The lattice agreement problem in crash failure model has been studied both in synchronous and asynchronous systems. 

In synchronous systems, a $\log n$ rounds recursive algorithm based on ``branch-and-bound'' approach is proposed by Attiya et al~\cite{attiya1995atomic} with message complexity of $O(n^2)$. The basic idea of their algorithm is to divide processes into two groups based on ids and let processes in the first group send values to processes in the second group. Each process in the second group takes join of received values. Then, this procedure continues within each subgroup. Their algorithm can tolerate at most $n - 1$ process failures. Later, the paper by Mavronicolasa et al~\cite{mavronicolasabound} gave an algorithm with round complexity of $\min \{1 + h(L), \lfloor (3 + \sqrt{8f + 1}/2) \rfloor \}$, for any execution where at most $f < n$ processes may crash and $h(L)$ denotes the height of the input lattice $L$. Their algorithm has the early-stopping property and is the first algorithm with round complexity that depends on the actual height of the input lattice. The best upper bound for the lattice agreement problem in crash-failure model is given by Xiong et al~\cite{zheng2018lattice}, which is $O(\log f)$ rounds. The basic idea of their algorithm is again to divide processes into two groups and trying to achieve agreement within each group recursively. But their criterion for dividing groups is based on the height of a value in the input lattice instead of process ids as in \cite{attiya1995atomic}. 

In asynchronous systems, the lattice agreement problem was first studied by Faleiro et al in \cite{faleiro2012generalized}. They present a Paxos style protocol when a majority of processes are correct. Their algorithm needs $O(n)$ asynchronous rounds in the worst case. The basic idea of their algorithm is to let each process repeatedly broadcast its current value to all at each round and update its value to be the join of all received values until all received values at a certain round are comparable with its current value. Later, Xiong et al~\cite{zheng2018lattice} proposes an algorithm which improves the round complexity to $O(f)$ rounds. For the best upper bound, a different paper by Xiong et al~\cite{zheng2018linearizable} presents an algorithm for this problem with round complexity of $O(\log f)$, which applies similar idea as~\cite{zheng2018lattice} but with extra work to take care of possible arbitrary delay of messages in asynchronous systems.  

The Byzantine failure model was first considered by Lamport et al~\cite{lamport1982byzantine} during the study of the Byzantine general problem. For the lattice agreement problem in Byzantine failure model, Nowak et al~\cite{nowak2019byzantine} gives an algorithm for a variant of the lattice agreement problem on cycle-free lattices that tolerates up to $f < \frac{n}{(h(X) + 1)}$ Byzantine faults, where $h(X)$ is the height of the input lattice $X$. In their problem, the original Downward-Vality and Upward-Validity requirement are replaced with a different validity definition, which only requires that for each output value $y$ of a correct process, there must be some input value $x$ of a correct process such that $x \leq y$. With their validity definition, however,  corresponding algorithms are not suitable for applications in atomic snapshot objects and linearizable replicated state machines, since each process would like to have their proposal value included its output value. A more closely related work is the preprint by Di Luna et al~\cite{di2019byzantine}, which proposes a reasonable validity condition and presents the first algorithm for asynchronous systems. Their algorithm takes $O(f)$ rounds. The basic idea of their algorithm is to first use the asynchronous Byzantine reliable broadcast primitive \cite{bracha1987asynchronous, srikanth1987simulating} to let all correct processes disclose their input values to each other, based on which each correct process constructs a set of safe values. This set of safe values are the only values a correct process will possibly deliver in future rounds. After the disclosure phase, the remaining steps are similar to the algorithms given in \cite{faleiro2012generalized, zheng2018lattice} for the lattice agreement problem in crash failure model, except that each process delivers a message only if all the values included in it are contained in its safe value set. Both the above two papers consider asynchronous systems. For synchronous systems with Byzantine failures, however, no algorithms have been proposed, which is the primary goal of this paper. 

For related works on application of lattice agreement, Faleiro et al~\cite{faleiro2012generalized} gives procedures to build a linearizable and serializable replicated state machine which only supports query operation and update operation but not mixed query-update operation, based on lattice agreement protocols. Later, Xiong et al~\cite{zheng2018linearizable} proposes some optimizations for their procedure for implementing replicated state machines in practice, specifically, they proposed a method to truncate the logs maintained. They also implemented a simple replicated state machine using their optimized protocols and compare it with S-Paxos \cite{biely2012s}, which is a java implementation of Paxos \cite{lamport1998part, lamport2001paxos}, and demonstrate better throughput and latency in some settings. The recent paper by Skrzypczak et al~\cite{skrzypczak2019linearizable} improves the procedure given in \cite{faleiro2012generalized} in terms of memory consumption, at the expense of progress, and high throughput is demonstrated. 

In a concurrent work by Di Luna et al~\cite{di2020synchronous}, they show that the Byzantine lattice agreement problem cannot be solved with $f \geq \frac{n}{3}$ failures in a synchronous systems. This shows that our algorithms achieve optimal resilience. In their paper, they give an algorithm for this problem which takes $O(\log f)$ rounds and tolerates $f < \frac{n}{4}$ failures. With the assumption of digital signatures, they can improve the resilience to be $f < \frac{n}{3}$. 

\section{System Model and Problem Definition}
 \subsection{System Model}
We assume a distributed message passing system with $n$ processes in a completely connected topology, denoted as $p_1,...,p_n$. Every process can send messages to every other process. We consider synchronous systems, which means that message delays and the duration of the operations performed by the process have an upper bound on the time. We assume that processes can have Byzantine failures but at most $f < n/3$ processes can be Byzantine in any execution of the algorithm. We use parameter $t$ to denote the actual number of Byzantine processes in a system. By our assumption, we must have $t \leq f$. Byzantine processes can deviate arbitrarily from the algorithm. We say a process is correct or non-faulty if it is not a Byzantine process.  We assume that the underlying communication system is reliable but the message channel may not be FIFO. In our algorithms, when a process sends a message to all, it also sends the message to itself. 

\subsection{The Byzantine Lattice Agreement Problem}
Let ($X$, $\leq$, $\sqcup$) be a finite join semi-lattice with the partial order $\leq$ and the join operation $\sqcup$. Two values $u$ and $v$ in $X$ are comparable iff $u \leq v$ or $v \leq u$. The join of $u$ and $v$ is denoted as $\sqcup \{u, v\}$. $X$ is a \textit{join semi-lattice} if a join exists for every nonempty finite subset of $X$. As customary in this area, we use the term {\em lattice} instead of
{\em join semi-lattice} in this paper for simplicity.  

In the Byzantine lattice agreement (BLA) problem, each process $p_i$ can propose a value $x_i$ in $X$ and must decide on some output $y_i$ also in $X$ with the presence of at most $f$ Byzantine processes in the system. Let $C$ denote the set of correct processes. Let $t$ denote the actual number of Byzantine processes in the system. An algorithm is said to solve the Byzantine lattice agreement problem if the following properties are satisfied: 

\textbf{Comparability}: For all $i \in C$ and $j 
\in C$, either $y_i \leq y_j$ or $y_j \leq y_i$.

\textbf{Downward-Validity}: For all $i \in C$, $x_i \leq y_i$. 

\textbf{Upward-Validity}: $\sqcup \{y_i ~ | ~ i \in C\} \leq \sqcup(\{x_i ~ | ~ i \in C\} \cup B)$, where $B \subset X $ and $|B| \leq t$.  \\

\begin{remark} 
The Upward-Validity given by Attiya et al~\cite{attiya1995atomic} is not suitable anymore with the presence of Byzantine processes, since the input value for a Byzantine process is not defined. Thus, the extra $B$ set is used to accommodate for possible values from Byzantine processes. The above Upward-Validity is similar to the Non-Triviality defined in \cite{di2019byzantine}. The only difference is that the extra $B$ set in \cite{di2019byzantine} is required to have size at most $f$, which is the resilience parameter. 
\end{remark}
In this paper, for a given set $V \subseteq X$, we use $\mathcal{L}(V)$ to denote the join-closed subset of $X$ that includes all elements in $V$. Clearly, $\mathcal{L}(V)$ is also a join semi-lattice. The height of a value $v$ in a lattice $X$ is defined as the length of longest chain from any minimal value to $v$, denoted as $h_X(v)$ or $h(v)$ when it is clear. The height of a lattice $X$ is the height of its largest value in this lattice, denoted as $h(X)$. For two lattices $\mathcal{L}_1$ and $\mathcal{L}_2$, we use $\mathcal{L}_1 \subseteq \mathcal{L}_2$ to mean that $\mathcal{L}_1$ is a sublattice of $\mathcal{L}_2$.

\section{$O(\sqrt{f})$ Rounds Algorithm for the BLA Problem}
In this section, we present an $O(\sqrt{f})$ algorithm for the BLA problem, which applies a slightly modified version of the Gradecast algorithm given by Feldman et al~\cite{feldman1988optimal} as a building block. 
Our algorithm has two primary ingredients which are quite different from the algorithm given in \cite{faleiro2012generalized, zheng2018lattice} for the crash failure model. In the Byzantine failure model, correct processes can receive arbitrary values from a Byzantine process. In order to guarantee {\bf Upward-Validity}, we do not want correct processes to accept arbitrarily many values sent from a Byzantine process. The idea in \cite{nowak2019byzantine} is to construct a safe value set, which stores the values reliably broadcast by each process at the first round. Later on, each process only delivers a received message if the values included in this message are contained in its safe value set. In this way, correct processes would not deliver arbitrary values sent by Byzantine processes. However, this idea can only provides $O(f)$ rounds guarantee.

To obtain the $O(\sqrt{f})$ rounds guarantee, our first idea is to let each correct process in our algorithm keep track of a lattice, which we call the safe lattice, instead of just a set of values. At each round, each correct process ignore all values received which are not contained in this safe lattice. By carefully updating this safe lattice of each correct process, our algorithm ensures that the value sent from a correct process is always in the safe lattice of any other correct process and Byzantine process cannot introduce arbitrary values to break the {\bf Upward-Validity} condition. To get the $O(\sqrt{f})$ bound, another crucial ingredient of our algorithm is to apply the Gradecast algorithm at each round to detect the Byzantine processes which sends different values to different correct processes and let each correct process ignores messages from these processes. This idea is used in \cite{ben2010simple} to solve the Byzantine consensus problem in synchronous systems.

\subsection{The Modified Gradecast Algorithm} 
Let us now first look at the modified Gradecast algorithm. Gradecast \cite{feldman1988optimal} is a three-round distributed algorithm that ensures some properties that are similar to those of broadcast. Specifically, in Gradecast there is a leader process $p$ that sends a value $v$ to all other processes. The output of process $i$ is a triple $<p, v^i_p, c^i_p>$ where $v^i_p$ is the value process $i$ thinks the leader $p$ has sent and $c^i_p$ is the score assigned by $i$ for the leader. A higher score indicates a higher confidence regarding the correctness of the leader. We say $c^i_p$ is the score assigned to the value or the leader by process $i$. The original Gradecast algorithm in \cite{feldman1988optimal} has the following properties.

1. If the leader $p$ is non-faulty then $v^i_p$ = $v$ and $c^i_p = 2$, for any non-faulty $i$;

2. For every non-faulty  $i$ and $j$: if $c^i_p > 0$ and $c^j_p > 0$ then $v^i_p = v^j_p$;

3. $|c^i_p - c^j_p| \leq 1$ for every non-faulty $i$ and $j$.

For our purpose, we do a slight modification of the Gradecast algorithm from \cite{feldman1988optimal} to enable processes to filter out some invalid values received. For completeness, we present the modified Gradecast algorithm in Fig \ref{fig:gradecast}. The only modification we do is to let each process store a lattice and filter out all received values which are not in the lattice. We call this lattice --- {\em the safe lattice}. 

\begin{figure} [htbp] \centering
\fbox{\begin{minipage}[t]  {4.5in}
\underline{{\bf Gradecast($q$, $v_q$)}:} \\
Each process $p$ keeps track of a safe value set $SV_p$ and a safe lattice $\mathcal{L}(SV_p)$: the join-closed subsets of $X$ including $SV_p$.\\
Each process $p$ ignores values received which do not belong to $\mathcal{L}(SV_p)$. \\
Each process $p$ uses $B_p$ to store Byzantine processes known to $p$. \\
Each process $p$ ignores messages received from processes in $B_p$. \\

{\bf Round 1:} \\ 
Leader $q$ sends $v_q$ to all \\

Each process $p$ executes the following steps \\
{\bf Round 2:} \\
$p$ sends the value received from $q$ to all \\
Let $<j, v_j>$ denote $p$ received $v_j$ from $j$ \\
Let $maj$ be a value received the most among such values \\
Let $\#maj$ be the number of occurrences of $maj$\\ 

{\bf Round 3:} \\
{\bf if} $\#maj \geq n - f$ \\
\h $p$ sends $maj$ to all \\
{\bf endif} \\

Let $<j, v_j>$ denote that $p$ received $v_j$ from $j$ \\
Let $maj^{'}$ be a value received the most among such values\\ 
Let $\#maj^{'}$ be the number of occurrences of $maj'$ \\

/* Grading */\\
{\bf if} $\#maj^{'} \geq n - f$ \\
\h set $v_p:= maj'$ and $c_p := 2$ \\
{\bf elif} $\#maj' \geq f + 1$ \\
\h set $v_p:= maj'$ and $c_p:= 1$ \\
{\bf else} \\
\h set $v_p := \bot $ and $ c_p := 0$ \\
{\bf endif}\\
Return $<q, v_p, c_p> $
\end{minipage} 
}
\caption{Gradecast Algorithm \label{fig:gradecast}}
\end{figure}

Each process $i$ keeps a safe value set, denoted as $SV_i$. This set is updated by process $i$ at each round of the main algorithm. From $SV_i$, each process $i$ constructs a safe lattice as the join-closed subset of $X$ which includes all values in $SV_i$, i.e, $\mathcal{L}(SV_i)$. The Gradecast procedure has two parameters. The first one specifies the leader of the Gradecast and the second one represents the value that the leader would like to send. When the leader invokes the Gradecast algorithm, all correct processes participate in a correct manner. While executing the Gradecast algorithm, each process $i$ ignores messages sent by processes in its bad set and ignores all received values which are not contained in its safe lattice $\mathcal{L}(SV_i)$. At round 1 of Gradecast, the leader broadcasts its value to all. At round 2, each process sends the value received from the leader if it is valid to all. At round 3, each process computes the most frequent value received at round 2 and sends this value to all if its frequency is at least $n - f$. Then, each process performs the grading step. It again computes the most frequent value received at round 3. If this value has frequency at least $n - f$, then it grades this leader with score 2 and sets the corresponding value to be the most frequent one. Otherwise, if the frequency is at least $f + 1$, it grades the leader with score 1 and sets the corresponding value to be the most frequent value received at round 3. Otherwise, it grades the leader with score 0 and sets the value to null (denoted by $\bot$). 

First assume that a correct leader always gradecast some value which lies in the safe lattice of each correct process and the bad set of each correct process does not contain any correct process. We will show that these assumptions are satisfied when we invoke the modified Gradecast algorithm as a substep in our main algorithm. 

\begin{lemma} \label{lem:gradecast_property}
Assume that a correct leader always gradecast some value which lies in the safe lattice of each other correct process and the bad set of each correct process does not contain a correct process. Then the modified Gradecast algorithm satisfies all properties of Gradecast. 
\begin{proof}
Proof of property 2 and 3 is the same as the proof of the original Gradecast algorithm. For property 1, since we assume that the value of a correct process is in the safe lattice of each other correct process and a correct process is not in any bad set, then no correct process would ignore its value. Therefore, in the grading phase, each correct process will grade a correct leader with score 2 and deliver the same value from the correct leader. 
\end{proof}
\end{lemma}

The main algorithm, shown in Fig \ref{fig:main}, runs in synchronous rounds. For ease of presentation, we use rounds to mean the iterations in the {\bf for} loop  of the main algorithm. We call the rounds taken by the Gradecast step as sub-rounds. In the main algorithm, each process $i$ keeps track of a bad set $B_i$, which contains all the Byzantine processes process $i$ has known. Each process which is graded with score at most 1 by process $i$ in some Gradecast step is included into $B_i$. Each process also keeps track of a set of safe values $SV_i$ from which the safe lattice is constructed. Initially, $SV_i$ is empty and each process accepts any received values at the first round. We assume for now that there is an upper bound $F$ on the number of rounds of the main algorithm. We will establish the accurate value of $F$ when we analyze the round complexity. 

\begin{figure} [htb]
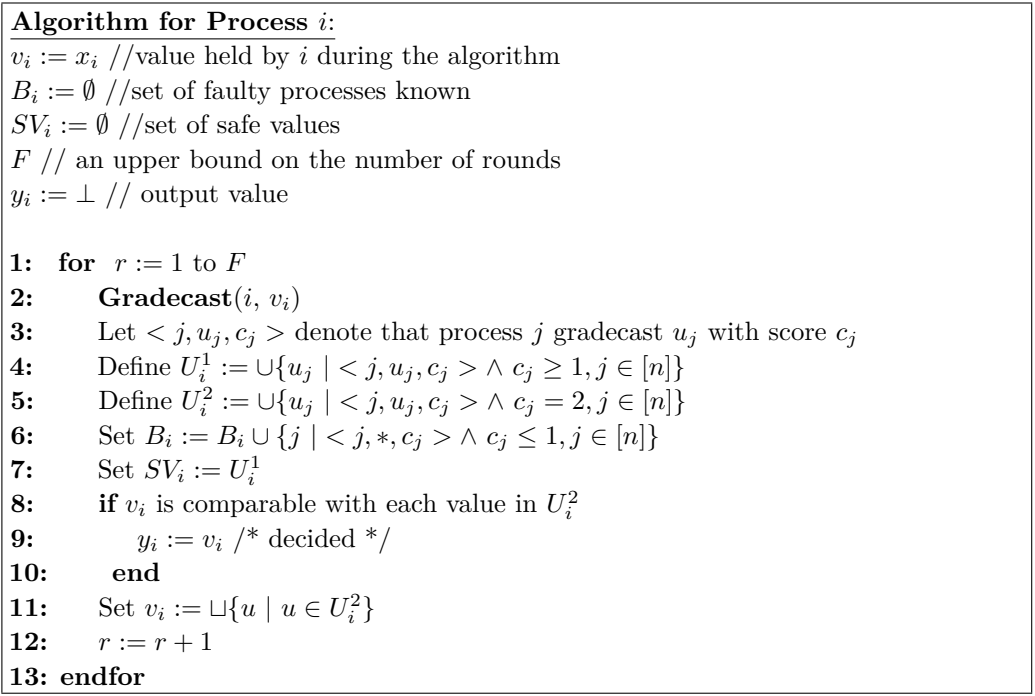
 \centering \label{fig:main}
\fbox{\begin{minipage}[t]  {5.2in}
\underline{{\bf Algorithm for Process $i$}:} \\
$v_i := x_i$ //value held by $i$ during the algorithm\\
$B_i := \emptyset$ //set of faulty processes known \\
$SV_i := \emptyset$ //set of safe values\\
$F$ // an upper bound on the number of rounds \\
$y_i := \bot$ // output value \\

{\bf 1: } {\bf for } $r := 1$ to $F$  \\
{\bf 2: } \h {\bf Gradecast}($i$, $v_i$) \\
{\bf 3: } \h Let $<j, u_j, c_j>$ denote that process $j$ gradecast $u_j$ with score $c_j$ \\
{\bf 4: } \h Define $U_i^1 := \cup \{u_j ~| <j, u_j, c_j> \wedge ~c_j \geq 1, j \in [n]\}$ \\
{\bf 5: } \h Define $U_i^2 := \cup \{u_j~ | <j, u_j, c_j> \wedge ~c_j = 2, j \in [n]\}$ \\
{\bf 6: } \h Set $B_i := B_i \cup \{j ~| <j,*,c_j> \wedge ~c_j \leq 1, j \in [n]\}$ \\
{\bf 7: } \h Set $SV_i := U_i^1$ \\
{\bf 8: } \h {\bf if} $v_i$ is comparable with each value in $U_i^2$\\
{\bf 9: } \h\h $y_i := v_i$ /* decided */\\
{\bf 10: } \h {\bf end} \\
{\bf 11:} \h Set $v_i :=  \sqcup \{u ~|~ u \in U_i^2\}$ \\
{\bf 12:} \h $r := r + 1$ \\
{\bf 13:} {\bf endfor} 
\end{minipage} 
}
\caption{$O(\sqrt{f})$ Rounds Algorithm for the BLA Problem \label{fig:main}}
\end{figure}

At each round, each process invokes the Gradecast algorithm with its current value and acts as the leader. So there are at least $n - f$ Gradecast instances running at each round, with each instance corresponding to one correct process. After the Gradecast phase, each process $i$ has a set of triples, one for each process which invoked Gradecast as the leader. A triple consists of the leader id, the value sent by the leader, and the score assigned by $i$. From these triples, processes $i$ updates its bad set $B_i$ and safe value set $SV_i$ as follows. At line 6, process $i$ includes all processes which are assigned score at most one into its bad set $B_i$ and ignores all messages sent from processes in $B_i$ at future rounds. Process $i$ also updates its safe value set $SV_i$ to be the union of all values gradecast by processes with score at least one. By updating the safe value set in such way, we can ensure that the current value of a correct process must be in the safe lattice of any other correct process. Thus, the value gradecast by a correct process in the next round is valid for any other correct process, which implies property 1 of Gradecast. On the other hand, this safe value set also prevents Byzantine processes from gradecasting an arbitrary value, i.e, Byzantine processes can only gradecast values that belong to the safe lattice. We explain this in details in our proof for correctness.

For the deciding condition, each process decides on its current value at a certain round if all values gradecast by processes with score 2 are comparable with its current value. A process keeps executing the algorithm even if it has decided. It updates its current value to be the join of all values gradecast with score 2 and starts the next round. 

\begin{table*}[h]
\caption{Notations}
\centering
\begin{tabular}{ |l|l|} 
 \hline
 \textbf{Variable} & \textbf{Definition} \\ \hline
 $C$ & The set of correct processes \\ \hline 
 $B_i^r$ & The set of Byzantine processes known by process $i$ at the end of round $r$ \\ \hline 
 $v_i^r$& Value of process $i$ at the end of round $r$ \\ \hline
 $v^r$& \makecell[l]{Auxiliary variable. The join of all values in $V^r$, i.e., $v^r := \sqcup \{u ~|~ u \in V^r\}$} \\ \hline
 $SV_i^r$ & The safe value set held by process $i$ at the end of round $r$ \\ \hline 
 $S^r$ & \makecell[l]{Auxiliary variable. The union of all safe value sets held by correct processes \\ at the end of round $r$, i.e, $S^r = \cup \{SV^r_{i} ~|~ i \in C\}$} \\ \hline 
 $s^r$ & Auxiliary variable. The join of all values in $S^r$, i.e, $s^r = \sqcup\{v ~|~ v \in S^r \}$\\ \hline $c^i_j$ & The score process $i$ assigned to process $j$ in the Gradecast with $j$ as the leader. \\\hline 
\end{tabular}
\label{tab:variable_sqrtf}
\end{table*}

\subsection{Proof of Correctness}
We now prove the correctness of our algorithm. The variables used for the proof are defined in Table \ref{tab:variable_sqrtf}. The main algorithm has the following properties. Property $(p1)$ and $(p8)$ immediately justify the assumption in Lemma \ref{lem:gradecast_property}. Thus, all properties of Gradecast are satisfied.

\begin{lemma} \label{lem:properties}
Let $i$ and $j$ be any two correct processes. For any round $1 \leq r \leq F$, the main algorithm satisfies the following properties. \\
(p1) $v^r_i \in \mathcal{L}(SV^r_j)$ \\
(p2) $v^r \in \mathcal{L}(SV^r_i)$ \\
(p3) $v^r \leq s^r$ \\
(p4) $v^{r - 1}_i < v^{r}_i$ if process $i$ is undecided at the end of round $r$ \\
(p5) $v^{r} < v^{r + 1}$ if at least one correct process is not decided at the end of round $r$ \\
(p6) $\mathcal{L}(S^{r + 1}) \subseteq \mathcal{L}(S^r)$ \\
(p7) $s^{r + 1} \leq s^r$ \\
(p8) For each correct process $i$, its bad set $B_i$ never contains a correct process 
\begin{proof}
$(p1)$: By the algorithm, $v^r_i$ is the join of all values gradecast with score 2 at round $r$. Consider an arbitrary value $v$ gradecast by some process $t$ with assigned score 2 by $i$ at round $r$, i.e, $c^i_t = 2$. Then, by property 3 of Gradecast, we have $c^j_t \geq 1$. Hence, from line 6 of the main algorithm, we have $v \in SV^r_j$. Since $v$ is some arbitrary value, we have $v^r_i \in \mathcal{L}(SV^r_j)$. \\
$(p2)$: Implied by $v^r = \{v_j^r ~|~ j \in C\}$ and $v_j^r \in \mathcal{L}(SV_j^r)$ by $(p1)$.  \\
$(p3)$: Since $\mathcal{L}(SV_j^r) \subseteq \mathcal{L}(S^r)$ and $s^r$ is the largest element in $\mathcal{L}(S^r)$, we have $v^r \in \mathcal{L}(S^r)$ and $v^r \leq s^r$. \\
$(p4)$: By the deciding condition, if $i$ is not decided at the end of round $r$, then it must receive at least one value which is incomparable with its current value. After taking joins, its new value must be greater than its old value. \\
$(p5)$: Immediately implied by $(p4)$. \\
$(p6)$: Suppose $\mathcal{L}(S^{r + 1}) \not \subseteq \mathcal{L}(S^r)$, then there must be some value $v$ gradecast by some process $b$ at round $r + 1$ such that $v \in \mathcal{L}(S^{r + 1}) \wedge v \not \in \mathcal{L}(S^r)$. Then, we have $v \not \in \mathcal{L}(SV^r_i)$ for any correct $i$. However, in the  Gradecast step at round $r$, for any correct process $i$, all received values which are not in $\mathcal{L}(SV^{r}_i)$ are ignored. Thus, each correct process grades value $v$ sent by $b$ with score 0. Thus, no correct process will include $v$ into their safe value set. Therefore, $v \not \in SV^{r + 1}_i$ for any correct $i$, a contradiction to $v \in \mathcal{L}(S^{r + 1})$.  \\
$(p7)$: Immediately implied by $(p6)$. \\
$(p8)$: Proof by induction on the round number $r$. \\
{\em Base case} ($r = 1$). At the first round, the safe lattice of each correct process is empty and each process accepts any message from any process. Then, by the Gradecast algorithm, the value of a correct leader will be graded with score 2. Thus, the bad set of each correct process does no contain a correct process. \\
{\em Induction case}: Assume that the property holds for all rounds before round $r$. At round $r$, since each correct leader is not in the bad set of any correct process, its value will be graded with score 2. Thus, it will not be included in a bad set of any correct process. 
\end{proof}
\end{lemma}

The following lemma is applied to show {\bf Upward-Validity}. 
\begin{lemma} \label{lem:at_most_one}
For any round $r$, the set $S^r$  contains at most one value gradecast by each process $i$ (possibly Byzantine) at round $r$. 
\begin{proof}
For a correct process $i$, it gradecasts a single value to all processes with score 2. Thus, the set $S^r$ contains exactly one value gradecast by it at round $r$. 

For a Byzantine process $i$, suppose $S^r$ contains two different values $v_i$ and $u_i$ gradecast by process $i$ at round $r$. Then there are two correct processes $p$ and $p$ such that $p$ includes $v_i$ into $SV^r_p$ and $q$ includes $u_i$ into $SV^r_q$. From line 6 of the algorithm, we have $c_p^i \geq 1$ and $c_q^i \geq 1$ for two correct processes $p$ and $q$ at round $r$. Then the fact that $v_i \not = u_i$ contradicts property 2 of Gradecast. 
\end{proof}
\end{lemma}



Now we show the output values obtained from the algorithm satisfy all the properties of lattice agreement. For now we assume that each process decides within $F$ rounds. We will obtain the accurate value of $F$ when we analyze the round complexity. 

\begin{lemma}
The values decided by correct processes satisfy all the properties defined in the BLA problem. 
\begin{proof}
{\bf Downward-validity}. Immediate from property $(p4)$ of lemma \ref{lem:properties}. 

{\bf Comparability}. We show that for any two correct processes $i$ and $j$, either $y_i \leq y_j$ or $y_j \leq y_i$. Let processes $i$ and $j$ decide at the rounds $r_i$ and $r_j$ respectively. Assume $r_i \leq r_j$, without loss of generality. Consider the following two cases. \\
Case 1: $r_i = r_j$. By the deciding condition, $y_i$ must comparable with $y_j$. \\
Case 2: $r_i < r_j$. Ar round $r_i$, the decided value of $i$ must be received by process $j$ with score 2, by property 1 of Gradecast. Since $j$ does not decide at this round, it must take join of the decided value of $i$. Thus, $y_i \leq y_j$.

{\bf Upward-Validity}. From property $(p3)$ and $(p7)$ of lemma \ref{lem:properties}, we have that $\sqcup \{y_i ~| ~ i \in C\} \leq v^F \leq s^F \leq s^1$. To bound $s^1$, from Lemma \ref{lem:at_most_one}, we know that $S^1$ contains at most $t$ extra values not from correct processes, since there are at most $t$ Byzantine processes. Then, we have $s^1 \leq \sqcup (\{x_i ~ | ~ i \in C\} \cup B)$ for some $B \subset X$ and $|B| \leq t$. Hence, $\sqcup \{y_i ~| ~ i \in C\} \leq \sqcup (\{x_i ~ | ~ i \in C\} \cup B)$ for some $B \subset X$ and $|B| \leq t$. 
\end{proof}
\end{lemma}

\subsection{Complexity Analysis}
We now analyze the complexity of our algorithm. We first show that the value of a correct process at a later round is at least the join of all values of correct processes in a previous round. 
\begin{lemma} \label{cla:correct_contains}
For any round $r$ and correct process $i$, $v^r \leq v^{r + 1}_i$. 
\begin{proof}
Since $v^r = \sqcup \{v^r_{j} ~| ~  j \in C\}$, each $v^r_j$ is gradecast by correct process $j$ to all at round $r + 1$ and by property 1 of Gradecast, process $i$ grades $j$ with score 2. Hence, process $i$ will take join of $v^r_j$ for any correct $j$, which implies $v^r \leq v^{r + 1}_i$.
\end{proof}
\end{lemma}

The following lemma along with property $(p5)$ and $(p7)$ of lemma \ref{lem:properties} guarantees the termination of our algorithm. 

\begin{lemma} \label{dec_2round}
If $v^r = s^r$ at the end of some round $r$, then all undecided correct processes decide in at most 2 rounds. 
\begin{proof}
Consider round $r + 1$, for each correct $i$, we have $v^r \leq v^{r + 1}_i$ from lemma \ref{cla:correct_contains} and $v^{r + 1}_i \leq s^{r + 1} \leq s^r$ from property $(p1)$ and $(p7)$. Since $v^r = s^r$, we must have $v^r = v^{r + 1}_i = s^r$ for each correct $i$. Then, at round $r + 2$, each valid value received by $i$ must be less than or equal to its current value, since its current value $v^r = s^r$ is the largest value in the whole safe lattice. Thus, the deciding condition is satisfied.
\end{proof}
\end{lemma}

\begin{lemma} \label{round_x}
$F \leq h(X) + 2$, where $X$ is the input lattice and $h(X)$ is the height of $X$. 
\begin{proof}
From property $(p7)$ of lemma \ref{lem:properties}, the height of $s^r$ in lattice $X$, is non-increasing as $r$ increases, whereas from property $(p5)$, the height of $v^r$ in lattice $X$ is strictly increasing if there is at least one undecided process as $r$ increases. Along with property $(p3)$ of lemma \ref{lem:properties}, which shows that $v^r \leq s^r$, we must have $v^r = s^r$ in at most $h(X)$ rounds. By lemma \ref{dec_2round}, after $h(X)$ rounds, all undecided processes decide in at most 2 more rounds. Thus, $F \leq h(X) + 2$.
\end{proof}
\end{lemma}

We now show the $O(\sqrt{f})$ rounds guarantee. The lemma below shows that if a bad process is included into the bad sets of all correct processes, its value would never be taken into consideration by correct processes in future rounds, i.e., these processes are completely ignored by all correct processes. 

\begin{lemma} \label{lem:bad_no_more}
For any process $b \in \cap \{B_i^{r - 1}~ | ~i \in C\}$, each correct process grades $b$ with score 0 in the Gradecast instance with $b$ as the leader at round $r$.  
\begin{proof}
In the Gradecast algorithm, all correct processes ignore values sent by processes in their bad sets. 
\end{proof}
\end{lemma}

We introduce the notion of terrible processes. 
 
\begin{definition}[terrible processes]
A process is terrible at round $r$ if it is graded with score 2 by at least one correct process in each round before $r$ and no correct process grades it with score 2 at round $r$.
\end{definition}

From the above definition, we observe that the terrible processes at each round are included into $B_i$ for each correct $i$ at line 6 of the algorithm. 

Let $\mathcal{L}^r$ denote the lattice formed by the safe value set at the end of round $r$, i.e, $\mathcal{L}^r = \mathcal{L}(S^r)$. Let $V^r$ denote the set of values which are assigned score of 2 by at least one correct process at round $r$. We observe that $v^r = \sqcup \{u ~|~ u \in V^r\}$, since each value in $V^r$ must be contained in the set $U_i^2$ for some correct process $i$ at line 5 of the algorithm. We also observe that $V^r \subseteq S^r$.
\begin{lemma} \label{lem:dec_fr} 
Suppose there are $f_r$ terrible processes at round $r$, then each process decides within $r + f_r + 2$ rounds. 
\begin{proof}
We first show that $h_{\mathcal{L}^r}(s^r) - h_{\mathcal{L}^r}(v^r) \leq f_r$ by showing that $|S^r - V^r| \leq f_r$. 
Consider round $r$ of the algorithm. Let $p$ be an arbitrary process. Consider the following three cases.

Case 1. $p \in \cap \{B_i^{r - 1}~ | ~i \in C\}$. From Lemma \ref{lem:bad_no_more}, we know that both $V^r$ and $SV^r$ do not contain a value from process $p$. 

Case 2. $p \not \in \cap \{B_i^{r - 1}~ | ~i \in C\}$ and $p$ is not terrible. In this case, process $p$ must gradecast a value $u$ and assigned score of 2 by at least one correct process. Due to property 3 of Gradecast, each correct process must assign score at least 1 for value $u$. Then, value $v$ must be contained in both $V^r$ and $S^r$. 

Case 3. $p$ is a terrible process. By Lemma \ref{lem:at_most_one}, we know that $p$ can only introduce at most one value into set $S^r$. 

From the above three cases, we can conclude that the values which are in $S^r$ but not in $V^r$ must be gradecast by terrible processes at round $r$. Thus, we have $|S^r - V^r| \leq f_r$. Since $v^r$ is the join of values in $V^r$ and $s^r$ is the join of values in $SV^r$, we must have $h_{\mathcal{L}^r}(s^r) - h_{\mathcal{L}^r}(v^r) \leq f_r$. 

On the other hand, from property $(p5)$ of lemma \ref{lem:properties}, we have $h_{\mathcal{L}^r}(v^{t + 1}) > h_{\mathcal{L}^r}(v^{t})$ for any round $t \geq r$. From property $(p7)$ of lemma \ref{lem:properties}, we have $h_{\mathcal{L}^r}(s^{t + 1}) \leq h_{\mathcal{L}^r}(s^{t})$ for any round $t \geq r$. Thus, we have that $h_{\mathcal{L}^r}(v^{r + k}) = h_{\mathcal{L}^r}(s^{r + k})$ for some $k \leq f_r$. Along with property $(p3)$, we have that $v^{r + k} = s^{r + k}$ for some number $k \leq f_r$. Then, by lemma \ref{dec_2round}, all undecided correct processes decide within at most 2 rounds. Therefore, each correct process decides within $r + f_r + 2$ rounds. 
\end{proof}
\end{lemma}

\begin{remark}
The crucial argument in Lemma \ref{lem:dec_fr} is that the height difference between $s^r$ and $v^r$ is at most $f_r$ in a round $r$ with at most $f_r$ terrible processes. If we simply store a safe value set instead of the safe lattice, this claim does not hold anymore.  
\end{remark}

\begin{lemma} \label{round_f}
$F \leq 2\sqrt{f} + 2$, where $f$ is the maximum number of Byzantine failures in the system such that $n \geq 3f + 1$.
\begin{proof}
Let us consider the first $\sqrt{f}$ rounds, at least one of these rounds has less than $\sqrt{f}$ terrible processes. By lemma \ref{lem:dec_fr}, starting from that round, each undecided process needs at most $\sqrt{f} + 2$ more rounds to decide. Thus, the total number of rounds is at most $2 \sqrt{f} + 2$. 
\end{proof}
\end{lemma}

\subsection{Early Stopping}
We say an algorithm has the early stopping property if its running time depends on the actual number of Byzantine processes during an execution of the algorithm. To obtain the early stopping property, we let each correct process dynamically update its termination round, which denotes the round number such that a correct process terminates the algorithm in any case. Let $t_i$ denote the termination round of the correct process $i$, initially set to be $F$. The value of $t_i$ is updated at each round in the following way. Consider a round $r$. Suppose $k^r_i$ processes are included into the bad set of $i$ at this round, then we set $t_i := \min \{t_i, r + k^r_i + 2\}$. Suppose we have $k$ actual Byzantine processes during an execution of the algorithm. Consider the first $\sqrt{k}$ rounds for an arbitrary correct process $i$, there must be a round $r'$ such that at most $\sqrt{k}$ processes are included into the bad set of $i$, then the termination round is at most $r' + \sqrt{k} + 2$. Thus, the termination round of $i$ is at most $2\sqrt{k} + 2$. Since the above argument applies for any correct process, all correct process must terminate in at most $2 \sqrt{k} + 2$ rounds.

\mainTheorem* 

\begin{proof}
 The round complexity follows from lemma \ref{round_x} and \ref{round_f} and the fact that the Gradecast step takes 3 sub-rounds. For message complexity, each Gradecast instance takes $O(n^2)$ messages. Since we have at most $n$ Gradecast instance at each round, each round takes $O(n^3)$ messages in total. Thus, the total number of messages is $O(n^3 \min\{h(X), \sqrt{f}\})$. 
By combining messages of different instances of Gradecast, the message
complexity for $n$ instances of Gradecast can be reduced from $O(n^3)$
to $O(n^2)$. Each message is a vector corresponding to $n$ instances of Gradecast. The message size can be reduced from $O(n)$ to $O(f)$ using
techniques based on error correcting codes as shown in \cite{DBLP:conf/sss/BridgmanG12}.
\end{proof}

\begin{corollary}
There is a $\min\{4h(X) + 8,8\sqrt{f} + 8\}$ rounds algorithm for the authenticated BLA problem in synchronous systems, which can tolerate $f < \frac{n}{2}$ Byzantine failures. The algorithm takes $O(n^2 \min\{h(X), \sqrt{f}\})$ messages. 
\begin{proof}
The paper \cite{katz2006expected} presents a 4-round Gradecast algorithm in the authenticated setting with $f$ failures such that $n \geq 2f + 1$, which also guarantees the three properties of gradecast. Then, we can modify this 4-round Gradecast algorithm in the authenticated setting in the same way as in the general setting. The main algorithm simply applies this modified 4-round Gradecast algorithm as a building block. The correctness proof remains the same. The time c omplexity is straightforward. 
\end{proof}
\end{corollary}

\section{$3\log n + 3$ Rounds Algorithm for BLA}
The $3\log n + 3$ round algorithm shown in this section is inspired by algorithms proposed for crash failure model in \cite{attiya1995atomic, zheng2018lattice}. The basic idea is to divide a group of processes into the slave subgroup and the master subgroup based on process ids, and ensure the property that the value of any correct process in the slave group is at most the value of any correct process in the master group. With the above property, if we recurse within each subgroup, then all correct processes can obtain comparable values in $O(\log n)$ rounds. In the Byzantine failure model, however, simply ensuring the above property is not enough. For example, suppose we divided a group of processes $G$ into the slave group $S(G)$ and the master group $M(G)$ such that the above property is satisfied. Suppose there is a Byzantine process in $S(G)$, it might send a value to some correct process in $S(G)$ in a later round such that the value is not known by correct processes in $M(G)$. Then, a slave process might have value which is greater than some master process. 

In order to prevent such cases, our algorithm introduces two novel ideas. First, when we divide a group into the slave subgroup and the master subgroup, we apply the {\bf Gradecast} algorithm to guarantee that the value of a slave process is at most the value of a master process. The {\bf Gradecast} algorithm serves the same purpose as the {\em Classifier} procedure as given in \cite{zheng2018lattice,zheng2018linearizable}. A nice property of the {\bf Gradecast} algorithm is that if some correct process assigns score 2 for the value gradecast by the leader, then each other correct process assigns score at least 1 for this value. Suppose we let each process in a group gradecast its value. Let $U^2$ denote the set of values assigned score of 2 by some correct process. Let $U^1$ denote the set of values assigned score of at least 1 by some correct process. Then, we must have $U^2 \subseteq U^1$. If each process in the master group updates their value to $U^1$ and each process in the slave group updates their value to $U^2$, then it is guaranteed that all values of the slave group is a subset of the values of each master process. However, the above property is only guaranteed at the current recursion level. Suppose there is a Byzantine process in the slave group, then it can gradecast a new value, which is not contained in the value set of some master process, to correct processes in the slave group. Then, the above property does not hold anymore. We need to ensure that the values of all slave processes are always a subset of the values of each master process when the recursion within each subgroup continues. To achieve that, each process keeps track of a safe value set for each other process and regards any value received from that process but not in the safe value set as invalid. This is also different from the algorithm in last section, where each process just keeps track of one single safe lattice for all. These safe values sets are used to restrict what values a process in a slave group can send. If we can guarantee that the union of all safe value sets for processes in the slave group a subset of the value set of each master process, then we can ensure that the above property continues to hold.

\begin{figure}[htp]
\fbox{\begin{minipage}[t]  {6in}
\underline{{\bf SetGradecast($q$, $V$)}:} \\
$V$ is a set of unique values to be gradecast by leader $q$.\\
Each process $i$ keeps track of a safe array $S_i$ of size $n$ with $S_i[j]$ storing a safe value set for $j$. \\
Process $i$ regards a value received from $j$ but not in $S_i[j]$ as invalid. \\
Each process ignores a set received if the set contains duplicate values \\

{\bf Round 1:} \\ 
Leader $q$ sends $V$ to all \\

/* Each process $p$ executes round 2 and round 3 */ \\
{\bf Round 2:} \\
/* Echo */ \\
process $p$ sends the set received from $q$ if any to all with invalid values removed \\
Let $<j, V_j>$ denote that $p$ received $V_j$ from $j$ with invalid values removed \\
Define multiset $U := \bigcup\limits_{j = 1}^{n} V_j$\\
Process $p$ constructs a set $U'$ as follows \\
{\bf for} each unique value $v \in U$ \\
\h Let $v_f$ denote its frequency in $U$\\
\h Add $v$ into $U'$ if $v_f \geq n - f$ \\
{\bf endfor} \\

{\bf Round 3:} \\
Process $p$ sends $U'$ to all \\
Let $<j, V_j>$ represent that $p$ received $V_j$ from $j$ with invalid values removed \\
Let $U := \bigcup\limits_{j = 1}^{n} V_j$\\

/* Grading */\\
Let $R$ denote a set to be returned \\
{\bf for } each value $v \in U$ \\
\h Let $v_f$ denote its frequency in $U$ \\
\h {\bf if} $v_f \geq n - f$ \\
\h\h set $c^p_v := 2$, add $v$ into $R$\\
\h {\bf elif} $v_f \geq f + 1$ \\
\h\h $c^p_v := 1$, add $v$ into $R$\\
\h {\bf else} \\
\h\h $c^p_v := 0$ \\
\h {\bf endif}\\
{\bf endfor} \\

Let $C_p$ denote the map storing the score for each value $v \in R$, i.e., $C_p[v] = c^p_v$. \\
Return$<R, C_p> $
\end{minipage} 
}
\caption{The {\bf SetGradecast} Algorithm}
\label{fig:setgradecast}
\end{figure}

\subsection{The SetGradecast Algorithm}

In the $3\log n + 3$ rounds algorithm, a process needs to gradecast a set of values instead of just one single value. Thus, we first present the following {\bf SetGradecast} algorithm, shown in Fig. \ref{fig:setgradecast}, which is similar to the {\bf Gradecast} algorithm in the previous section, except that it is used to gradecast a set of values. In the $3\log n + 3$ rounds algorithm, each process $i$ keeps track of a safe array $S_i$ of size $n$ with $S_i[j]$ being a safe value set for process $j$. 

Process $i$ considers a value $v$ received from process $j$ as valid if $v \in S_i[j]$, otherwise invalid. Process $i$ uses $S_i$ to filter out invalid values received from any other process in the {\bf SetGradecast} algorithm. We will show how to construct and update the safe value array for each process in the main algorithm. 

In the {\bf SetGradecast} algorithm, we assume that the leader needs to gradecast a set of distinct values, which can be guaranteed by introducing some unique tags for each value. If a process receives a message which contains duplicate values from some leader, it just ignore the message. Similar to the {\bf Gradecast} algorithm, a process in the {\bf SetGradecast} algorithm needs to send any valid value received at round 2 and round 3 from at least $n - f$ different processes to all other processes. Each process $i$ returns a triple $<j, R_j, C_j>$ when process $j$ invokes the {\bf SetGradecast} as the leader. The set $R_j$ is the set of values gradecast by process $j$ with score at least 1 and the map $C_j$ stores the score assigned by process $i$ for each value in $R_j$. 

Let $v$ be an arbitrary value gradecast by the leader. Let $c^i_v$ denote the score of $v$ assigned by process $i$. 
\begin{lemma}
Algorithm {\bf SetGradecast} has the following properties. 

1. If a value $v$ gradecast by a correct leader $i$ is in the safe value set of each correct process $j$ for $i$, i.e., $v \in S_j[i]$ for each correct $j \in [n]$, then the score of $v$ assigned by each correct $j$ must be 2, i.e., $c^j_v = 2$.

2. Let $v$ be an arbitrary value gradecast by the leader. Then $|c^i_v - c^j_v| \leq 1$ for any two correct process $i$ and $j$.

3. If a value $v$ gradecast a leader $i$ is not in the safe value set of any correct process for $i$, i.e., $v \not \in S_j[i]$ for any $j \in C$, then $c^j_v = 0$ for each $j \in C$. 
\begin{proof}
For property 1, since the leader is correct, it must send the same value $v$ to each process at round 1. Since $v$ is contained in the safe value set of each correct process $j$ for the leader, then at round $2$ each correct process $j$ will receive $v$ from each correct process. It then sends $v$ to each process at round 3. At round 3, each correct $j$ receives $v$ from each correct process, thus the score of $v$ assigned by each correct $j$ must be 2. 

For property 2, we only need to show that if any correct $j$ has $c^j_v = 2$, then any other correct process $k$ must have $c^k_{v} \geq 1$. Since $j$ has $c^j_{v} = 2$, then process $j$ must have received $v$ from at least $n - f$ processes at round 3, which contains at least $n - 2f$ correct processes. Thus, for any other correct process $k$, it must have received $v$ from at least $n - 2f$ correct processes, which is at least $f + 1$ by the assumption of $n \geq 3f + 1$. Therefore, $c^k_{v} \geq 1$.  

For $v$ to be assigned score at least 1 by a correct process $j$, process $j$ has to receive this value from at least $f + 1$ different processes at round 3. Since each correct process $j$ considers any value from $i$ but not in $S_j[i]$ as invalid and there are at most $f$ Byzantine processes, any correct process can receive $v$ from at most $f$ different processes. Thus, it must be assigned score of 0.
\end{proof}
\end{lemma}

\begin{figure}[htp]
    \centering
    \fbox{\begin{minipage}[t]  {6in}
Initially all process are within the same group $G_1$ \\
Let $\mathcal{G}_r$ denote the collection of groups at round $r$, initially $\mathcal{G}_1 = \{G_1\}$ \\
/* Each process $i \in [n]$ stores the following variables */ \\
$x_i$: input value for process $i$\\ 
$V_i$: value set of process $i$ with $V_i := \{x_i\}$ initially. \\
$S_i$: an array of size $n$ with $S_i[j]$ being the safe value set for process $j$. \\

/* Build the Initial Safe Array */ \\
{\bf 1: } {\bf for} each process $i$, {\bf in parallel do} \\
{\bf 2: } \h Process $i$ invokes {\bf Gradecast}($i$, $x_i$) \\
{\bf 3: } \h Let $<j,v_j, c_j>$ denote the tuple obtained from the gradecast of $p_j$\\
{\bf 4: } \h Set $S_i[k] := \{v_j ~|~ c_j \geq 1, j \in [n]\}$ for each process $k \in [n]$\\
{\bf 5: } {\bf endfor} \\

{\bf 6: } {\bf for} $r := 1$ to $\log n$ \\
{\bf 7: } \h Divide each group $G \in \mathcal{G}_r$ into slave group $S(G)$ and master group $M(G)$ based on ids\\
{\bf 8: } \h Let $\mathcal{G}_{r + 1}$ denote the collection of new groups\\
{\bf 9: } \h Each slave process $p$ executes {\bf SetGradecast}($p$, $V_p$) \\
{\bf 10:} \h  {\bf for} each process $i$, {\bf in parallel do} \\
{\bf 11:} \h\h Let $<j, Val_j^i, C_j^i>$ denote the value-score tuple that $p_i$ received from the gradecast of $p_j$, \h\h\h~ with $Val_j^i$ being the values and $C_j^i$ being the score map\\ 
{\bf 12:} \h\h Let $R_j^i$ denote the set of values assigned score of 2 by $p_i$ in the gradecast of $p_j$, \\ \h\h\h~ i.e, $R_j^i := \{v \in Val_j^i ~ | ~ C_j^i[v] = 2\}$ \\
{\bf 13:} \h\h  {\bf for} each group $G \in \mathcal{G}_r$ \\
{\bf 14:} \h\h\h Let $U_1$ denote the set of values gradecast by processes in $S(G)$ with score at least 1 \\ \h\h\h \h ~i.e.,  $U_1 := \bigcup\limits_{j \in S(G)} Val_j^i$ \\
{\bf 15:} \h\h\h Let $U_2$ denote the set of values gradecast by processes in $S(G)$ with score 2\\ \h\h\h \h ~i.e., $U_2 := \bigcup\limits_{j \in S(G)} R_j^i$ \\
{\bf 16:} \h\h\h $S_i[j] := U_2$ for each process $j \in S(G)$ \\ 
{\bf 17:} \h\h\h $S_i[j] := S_i[j] \cup U_1$ for each process $j \in M(G)$,  \\ 
{\bf 18:} \h\h\h {\bf if} $i \in S(G)$ {\bf then} $V_i := U_2$ \\
{\bf 19:} \h\h\h {\bf elif} $i \in M(G)$ {\bf then} $V_i := U_1$\\
{\bf 20:} \h\h {\bf endfor} \\
{\bf 21:} \h {\bf endfor} \\
{\bf 22:} {\bf endfor} \\
{\bf 23:} Output $y_i := \sqcup \{v \in V_i\}$
\end{minipage} 
}
\caption{The $3\log n + 3$ Rounds Algorithm for the BLA Problem}
\label{fig:logn_algo}
\end{figure}

\subsection{The Main Algorithm}
The $3\log n + 3$ rounds algorithm is shown in Fig. \ref{fig:logn_algo}. In the algorithm, each process $i$ stores a value set $V_i$ which is updated at each round. Initially, $V_i = \{x_i\}$. Each process $i$ keeps track of a safe array $S_i$ of size $n$ with $S_i[j]$ being the safe value set for $j$. Process $i$ regards any value received from $j$ which is not in $S_i[j]$ as invalid and thus ignores it. Different processes may have different safe value set for a process $j$. Initially, all processes are in the same group, denoted as $G_1$. During the algorithm, processes might be divided into different groups. The algorithm proceeds as follows. 

The initial round from line 1 to line 5 is used to build the initial safe array of each process $i$. At this round, each process $i$ invokes the {\bf Gradecast} algorithm to send its input value $x_i$ to all. Each process $i$ constructs the same initial safe value set for each process $j$, which includes all values gradecast by some process and assigned score of at least 1 by process $i$. We will show later that this initial round of gradecast guarantees {\bf Upward-Validity} of the BLA problem. Intuitively, this is because each Byzantine process can only introduce one value into the safe value sets by properties of Gradecast. 

In lines 6-22, at each round $r$ from 1 to $\log n$, each group $G$ is divided into two subgroups: the slave group $S(G)$ and the master group $M(G)$, where $S(G)$ contains all processes in $G$ with ids in the lower half and $M(G)$ contains all processes in $G$ with ids in the upper half. Each process in $S(G)$ invokes {\bf SetGradecast} to gradecast its current value set to all processes. Processes in $M(G)$ do not gradecast their value sets. After this step, for each group $G$, process $i$ obtains a set of values gradecast by processes in $S(G)$. From line 13 to line 20, process $i$ updates its safe value set $S_i$ and its value set $V_i$ based the values obtained in all {\bf SetGradecast} instances. At line 16, process $i$ updates its safe value set for each process $j \in S(G)$ to be the values gradecast with score 2 by some process in $S(G)$. At line 17, it updates the safe value set for each process in $M(G)$ to be the values gradecast with score at least 1 by some process in $S(G)$. If process $i$ is a slave process in $S(G)$, it updates its value set to be the set of values gradecast by processes in $S(G)$ and assigned score of 2. If it is a master process in $M(G)$, it updates its value set to be the set of values gradecast by processes in $S(G)$ with score at least 1. 

Consider group $G$ at round $r$. Let $S(G)$ and $M(G)$ denote the slave group and the master group obtained when dividing $G$ at round $r$. Each master process in $M(G)$ updates its value set to be the set of values gradecast by processes in $S(G)$ and assigned score of at least 1 by itself. Each slave process in $S(G)$ updates its value set to be the set of values gradecast by processes in $S(G)$ and assigned score 2 by itself. By property of {\bf SetGradecast}, we can observe that at the end of round $r$, the value set of each correct slave process in $S(G)$ is a subset of the value set of each correct master process in $M(G)$. 
 
\newpage
\subsection{Correctness and Complexity}
Now we analyze the correctness and complexity of our algorithm. For any group $G$, let $S(G)$ and $M(G)$ denote the slave group and the master group obtained when dividing $G$. The variables we use for analysis are given in Table \ref{tab:variable_logn}.

\begin{table*}[h]
\caption{Notations}
\centering
\begin{tabular}{ |c|c|} 
 \hline
 \textbf{Variable} & \textbf{Definition} \\ \hline
$V_i^r$ & The value set held by process $i$ at the end of round $r$.  \\ \hline
$S_i^r$ & The safe value array of $i$ at the end of round $r$. \\ \hline 
$SF_j^r$ & \makecell{Auxiliary variable. The union of the safe value sets of all correct process  \\ for $j$ at the end of round $r$, i.e., $SF_j^r := \{S_i^r[j] ~| ~ i \in C\}$}\\ \hline
$SF_G^r$ & \makecell{Auxiliary variable. The union of the safe value sets of all correct process \\ for processes in group $G$ at the end of round $r$, i.e., $SF_G^r := \bigcup\limits_{j \in G} SF_j^{r}$} \\ \hline 
\end{tabular}
\label{tab:variable_logn}
\end{table*}


\begin{lemma} \label{lem:cls_logn}
Let $G$ be a group which is divided into $S(G)$ and $M(G)$ at round $r$. Then \\
(p1) $SF_{S(G)}^{r} \subseteq SF_{G}^{r - 1}$ \\
(p2) $SF_{M(G)}^{r} \subseteq SF_{G}^{r - 1}$ \\
(p3) For each correct process $i \in G$, $V_i^{r} \subseteq SF_G^{r - 1}$ \\
\begin{proof}
{\bf (p1):} At round $r$, processes in group $S(G)$ gradecast their values to all. Consider an arbitrary value $v \in SF^r_{S(G)}$. From line 16, we know that $v$ must be gradecast by some process in $S(G)$ and assigned score of 2 by at least one correct process. This implies that value $v$ must be in the safe value set of at least one correct processes for $j$ at the beginning of round $r$ by property 3 of {\bf SetGradecast}. Thus, value $v$ must be in $SF_G^{r - 1}$. 

{\bf (p2):} From line 17 of the algorithm, for an arbitrary process $j \in M(G)$, each correct process $i$ sets $S_i^r[j] := S_i^{r - 1}[j] \cup U_1$, with $U_1$ being the set of values gradecast by some process in $S(G)$ at round $r$ and assigned score of at least 1 by process $i$. Thus, any value $v \in U_1$ must be in the safe value set for $j$ of at least one correct process by property 3 of {\bf SetGradecast}. Then, $U_1$ must be in $SF_G^{r - 1}$. Therefore, $SF_{M(G)}^{r} \subseteq SF_{G}^{r - 1}$. 

{\bf (p3):} From line 18-19, we know that $V_i^r$ is either updated to be the set of values assigned score 2 or score at least 1 by process $i$ at round $r$. Then, by property 3 of {\bf SetGradecast}, $V_i^r \subseteq SF_G^{r - 1}$. 

\end{proof}
\end{lemma}

The following lemma shows that if a value in the value set of correct process $i$ is contained in the safe value set of each correct process for process $i$, then this value remains in the value set of process $i$. 
\begin{lemma} \label{lem:correct_value}
Consider an arbitrary value $v \in V^r_j$ of correct process $j$. If it is contained in $S^r_i[j]$ for each correct process $i$, then we have \\
(1) $v \in S^t_i[j]$ for each correct process $i$ and $t \geq r$. \\
(2) $v \in V^t_j$ for any $t \geq r$. 
\begin{proof}
(1). By induction on the round number. For the base case $t = r$, the claim holds. Consider round $t > r$ and assume that (1) and (2) are satisfied for any round before $t$. In the algorithm, $S_i[j]$ can only shrink at line 21 and $V_j$ can only shrink at line 24, which happens when process $j$ is divided into the slave group at round $t$. Then $S_i[j]$ is updated to be the set of values gradecast by some process in the group of $j$ and assigned score of 2 by process $i$. At round $t$, since $v \in V^{t - 1}_j$ by induction hypothesis and $j$ is slave at round $t$, process $j$ must gradecast value $v$. Since $v$ is in the safe value set of each correct $i$ by induction hypothesis, we must have $c^i_v = 2$  by property 1 of {\bf SetGradecast}. Thus, value $v$ must be included into $S_i[j]$ by each correct $i$ at line 16 and $v$ must included into $V^t_j$ at line 18 by process $j$. 
\end{proof}
\end{lemma}

The following lemma shows that the union of all safe value sets of correct processes for slave processes will always be a subset of the values of each master process. 
\begin{lemma} \label{lem:domination}
Let $G$ be a group which is divided into $S(G)$ and $M(G)$ at round $r$. Then $SF_{S(G)}^r \subseteq V_j^t$ for each correct $j \in M(G)$ and any round number $t \geq r$.   
\begin{proof}
Consider round $r$ of the algorithm. Consider an arbitrary value $v \in SF_{S(G)}^r$. The value $v$ must be gradecast by some process in $S(G)$ and assigned score 2 by some correct process by line 16 of the algorithm. That is, there exists a correct process $t$ with $c_v^t = 2$. Then, we have $c_v^k \geq 1$ for any correct process $k$ by property 2 of the {\bf SetGradecast} algorithm. Hence, $v \in S^r_k[j]$ for any correct process $k$ by line 17 of the algorithm. Also, we have $v \in V^r_j$ by line 19. By Lemma \ref{lem:correct_value}, we can derive that $v \in V_j^{t}$ for any round number $t \geq r$. Since $v$ is an arbitrary value from $SF_{S(G)}^r$, we must have $SF_{S(G)}^r \subseteq V_j^{t}$ for any round number $t \geq r$.
\end{proof}
\end{lemma}

\begin{lemma} \label{lem:comp}
(Comparability) For any two correct process $i$ and $j$, we have either $y_i \leq y_j$ or $y_j \leq y_i$. 
\begin{proof}
We show either $V_i^{\log n} \subseteq V_j^{\log n}$ or $V_j^{\log n} \subseteq V_i^{\log n}$. Let $G$ denote the last group both $i$ and $j$ belong to. Let $r$ be the round when $G$ is divided into $S(G)$ and $M(G)$. Without loss of generality, suppose $i \in S(G)$ and $j \in M(G)$. By repeatedly applying $(p1)$, $(p2)$ and $(p3)$ of Lemma \ref{lem:cls_logn}, we have $V_i^{\log n} \subseteq SF_{S(G)}^r$. By Lemma \ref{lem:domination}, we know that $SF_{S(G)}^r \subseteq V_j^{\log n}$. Therefore, $V_i^{\log n} \subseteq V_j^{\log n}$. 
\end{proof}
\end{lemma}

\logTheorem* 

\begin{proof}
{\bf Comparability} follows from Lemma \ref{lem:comp}. 

{\bf Downward-Validity.} After the end of the initial round, the input value $x_i$ of each correct process $i$ must be contained in its value set $V_i$ and in the safe value set of each process $j \in C$ for $i$, by property 1 of {\bf SetGradecast}. Then, by Lemma \ref{lem:correct_value}, the input value of $i$ remains in $V_i$. Thus, $x_i \leq y_i$. 

{\bf Upward-Validity.} After the initial round, since only values with score at least 1 are included into the safe value set, by property 2 of the {\bf Gradecast} algorithm, each Byzantine process can only introduce at most one value into the initial safe value set. Then, we have $SF_{G_1}^0 \subseteq \cup \{x_i ~ | ~ i \in C\} \cup B$, where $B \subset X $ and $|B| \leq t$. For each correct process $i$, by repeatedly applying $(p1)$, $(p2)$ and $(p3)$ of Lemma \ref{lem:cls_logn}, we have $V_i^{\log n} \subseteq SF_{G_1}^0$. Thus, $\sqcup \{y_i ~|~ i \in C\} \leq \sqcup(\{x_i ~ | ~ i \in C\} \cup B)$, where $B \subset X $ and $|B| \leq t$.
\end{proof}

\begin{corollary}
There is a $4\log n + 4$ rounds algorithm for the authenticated BLA problem in synchronous systems which can tolerate $f < \frac{n}{2}$ Byzantine failures, where $n$ is the number of processes in the system. The algorithm takes $O(n^2 \log n)$ messages. 
\end{corollary}

\newpage 

\section{$O(\log f)$ Algorithm for the BLA Problem}

In this section, we present an algorithm for the BLA problem which takes $4 \log f + 3$ synchronous rounds. The $O(\log f)$ algorithm in \cite{zheng2018lattice} for the crash failure model applies a classifier procedure to divide a group of processes into two subgroups: the master group and the slave group and guarantees the following two properties. 1) The value of any slave process is at most the value of any master process. 2) The union of all slave values is bounded by a threshold parameter $k$, where $k$ serves as a knowledge threshold. With the above properties, the classifier procedure can be recursively applied within each subgroup and all processes have comparable values after $O(\log f)$ rounds by setting the knowledge threshold $k$ in a binary way as follows. Initially, all processes are in the same group with initial knowledge threshold $n - \frac{f}{2}$. Consider a group $G$ at level $r$ with knowledge threshold $k$, then the slave group of $G$ has knowledge threshold $k - \frac{f}{2^{r + 1}}$ and the slave group of $G$ has knowledge threshold $k + \frac{f}{2^{r + 1}}$. If all processes exchange their values before recursively invoking the classifier procedure, we can see that each correct process must have at least $n - f$ values and have at most $n$ values. Then, after $\log f$ levels of recursion, by property 1 and 2, all processes in different groups must have comparable values and all processes in the same group must have the same value. 

The classifier procedure for the crash failure model in~\cite{zheng2018lattice} is quite simple. All processes within the same group exchange their current values. If a process obtains a value set with size greater than the threshold $k$, it is classified as a master. Otherwise, it is classified as a slave. A slave process keeps its value unchanged. A master process updates its value to be the union of all values obtained. Property 1 is straightforward. For property 2, since each slave process must have received all the values of slave processes, the union of all slave values must have size at most $k$. 

In presence of Byzantine processes, however, guaranteeing property 2 is quite challenging since some of the slave processes may be Byzantine. Similar to the $O(\log n)$ algorithm, we let each process store a safe value set for each group which bounds the values that processes in this group can send at each round. With these safe value sets, we strengthen property 2 as follows: 2) The union of the safe value sets of all correct processes for the slave group has size at most $k$.

 In the $O(\log n)$ algorithm, we divide a group into slave subgroup and master subgroup based on process ids. A Byzantine process cannot lie about its group identity, i.e., whether it is in the slave group or the master group. In the $O(\log f)$ algorithm, the classification is based on the values received at each round. So, process $i$ does not know whether process $j$ is classified as a slave or a master at any given round. Thus, a Byzantine process can lie about its group identity and the $O(\log f)$ algorithm needs a mechanism to prevent such lies. 
 In the $O(\log f)$ algorithm, each process $i \in [n]$ has a label $l_i$, which serves as the threshold when it performs the classification step and also indicates its group identity. The notion of group defined below is based on labels of processes. 
 
 \begin{definition}[group]
A {\em group} is a set of processes which have the same label. The {\em label of a group} is the label of the processes in this group. The label of a group is also the threshold value processes in this group use to do classification. 
\end{definition}

Initially all processes are within the same group with label $k_0 = n - \frac{f}{2}$. We also use label to indicate a group. We say a process is in group $k$ if its value is associated with label $k$.

Consider the classification step for group $G$ with label $k$. There are two main differences between $O(\log n)$ and
$O(\log f)$ algorithm. First, in $O(\log n)$ algorithm,  we let each process in the slave group of $G$ invoke the gradecast primitive to send its current value set. In $O(\log f)$ algorithm, all processes in group $G$ invoke the gradecast primitive. 
Second, in the $O(\log n)$ algorithm, the classification is based on process ids.
In the $O(\log f)$ algorithm, we let each process send its safe value set for group $G$ to all processes in $G$ and each process in $G$ perform classification based on the safe sets received. If the union of the received safe sets has size greater than $k$, the process is classified as a master, otherwise as a slave.
Each slave process updates its value to be the set of values with score 2. Each master process updates its value to be the union of its current value and the set of values with score at least 1. The safe value set is updated in a similar manner to the $O(\log n)$ algorithm. 

\begin{figure}[htp] 
    \centering
    \fbox{\begin{minipage}[t]  {6in}
$x_i$: input value for process $i$\\ 
$V_i$: value set of process $i$. $V_i := \{x_i\}$ initially. \\
Initially all process are within the same group with label $k_0 = n - \frac{f}{2}$ \\
$l_i$: label of $p_i$. $l_i := k_0 = n - \frac{f}{2}$ initially and is updated at each round \\
Map $F_i$, with $F_i[k]$ being the safe value set for processes with label $k$. \\

/* Build the Initial Safe Map and Value Set */ \\
{\bf 1: } {\bf for} each process $i$, {\bf in parallel do} \\
{\bf 2: } \h Execute {\bf Gradecast}($i$, $x_i$) \\
{\bf 3: } \h Let $<j, v_j, c_j>$ denote that process $j$ gradecasts $v_j$ with score $c_j$ \\
{\bf 4: } \h Set $F_i[k_0] := \{v_j ~ | ~ c_j \geq 1 \wedge j \in [n]\}$ // safe set for the initial group with label $k_0 = n - \frac{f}{2}$\\ 
{\bf 5: } \h Set $V_i := \{v_j ~ | ~ c_j = 2 \wedge j \in [n]\}$\\
{\bf 6: } {\bf endfor} \\

{\bf 7: } {\bf for} $r := 1$ to $\log f$ \\
{\bf 8: } \h  {\bf for} each process $i$, {\bf in parallel do} \\
{\bf 9: } \h \h Execute {\bf SetGradecast}$(i, V_i, l_i)$ \\ 
{\bf 10:} \h\h Let $L$ denote the set of labels received in all Gradecasts\\ 
{\bf 11:} \h\h {\bf for} each label $k \in L$ /* Each label represents a group */  \\
{\bf 12:} \h\h\h Let $U_{i,1}^k$ denote the set of values gradecast with label $k$ and assigned score $\geq$ 1 by $p_i$. \\
{\bf 13:} \h\h\h Let $U_{i,2}^k$ denote the set of values gradecast with label $k$ and assigned score 2 by $p_i$\\
{\bf 14:} \h\h\h Set $F_i[m(k,r)] := F_i[k] \cup U_{i,1}^k$ and $F_i[s(k,r)] := U_{i,2}^k$\\
{\bf 15:} \h\h\h Send $U_{i,2}^k$ to processes who gradecast with label $k$. \\
{\bf 16:} \h\h\h {\bf if} $l_i = k$ //if my label is $k$\\
{\bf 17:} \h\h\h\h Let $R_j$ denote the set of values received from $p_j$ at line 15\\
{\bf 18:} \h\h\h\h Set $T := \cup \{R_j ~|~ R_j \subseteq U_{i,1}^k, j \in [n]\}$ \\
{\bf 19:} \h\h\h\h /* Classification */ \\
{\bf 20:} \h\h\h\h {\bf if} {$|T| > k$} {\bf then} $V := U_{i,1}^k$, $l_i := l_i + \frac{f}{2^{r + 1}}$ //master process\\
{\bf 21:} \h\h\h\h\h\h\h\h {\bf else} $V := U_{i,2}^k$, $l_i := l_i - \frac{f}{2^{r + 1}}$ //slave process\\
{\bf 22:} \h\h\h {\bf endif} \\
{\bf 23:} \h\h{\bf endfor} \\
{\bf 24:} \h{\bf endfor} \\
{\bf 25:} {\bf endfor} \\
{\bf 26:}  $y_i := \sqcup \{v \in V_i^{\log f + 1}\}$
\end{minipage} 
}
\caption{The $4\log f + 3$ Rounds Algorithm for the BLA Problem}
\label{fig:logf_algo}
\end{figure}

 In the $O(\log f)$ algorithm, each process $i \in [n]$ keeps track of a safe value map $F_i$ with $F_i[k]$ being the safe value set of process $i$ for group $k$, i.e., $F_i[k]$ is an upper bound on the values with label $k$ that process $i$ considers valid. Define $s(k,r) = k - \frac{f}{2^{r + 1}}$ and $m(k, r) = k + \frac{f}{2^{r + 1}}$. The algorithm is shown in Fig. \ref{fig:logf_algo}. 
 
 In lines 1-6 of the algorithm, each process first invokes {\bf Gradecast} to send its input value to all. Then, it constructs its value set as the set of values gradecast with score 2 and its initial safe value set for the initial group as the set of values gradecast with score at least 1. Lines 1-6 serve two purposes: 1) To construct the safe value set for the initial group with label $k_0 = n - \frac{f}{2}$ and ensure that each Byzantine process can introduce at most one value in the safe value sets 2) Ensure that there are at most $f$ values unknown to each correct processes. Then,  $\log f$ recursion levels suffice for all processes to obtain comparable values. 

In lines 7-25, at each round, each process invokes {\bf SetGradecast} to send its current value to all and performs classification. When a process invokes the {\bf SetGradecast} to send its current value set, its current label is attached. After the gradecast step, for each process $j \in [n]$, process $i$ obtains a value-score-label triple from the gradecast of process $j$. Then, process $i$ executes line 11-15 for each group at the current round. Specifically, for the group with label $k$, process $i$ obtains the set of values with score at least 1, $U_{i,1}^k$, and the set of values with score 2, $U_{i,2}^k$,  from the gradecast of processes with label $k$. Then, process $i$ updates its safe value set for group $m(k,r)$, i.e., the master group of group $k$, to be the union of its safe value set for group $k$ and $U_{i,1}^k$. It also updates its safe value set for group $s(k,r)$, i.e., the slave group, to be $U_{i,2}^k$. Due to property 2 of gradecast, we must have $U_{i,2}^k \subseteq U_{j, 1}^k$ for any process $i$ and $j$. This step is to ensure that if a master process $j$ obtains a value in $U_{i,2}^k$, this value will be in the safe value set of each correct process for $j$. Then, when the master process tries to gradecast this value, it must be assigned score of 2 by each correct process, due to property 3 of gradecast. At line 15, process $i$ sends its safe value set for the slave group to the processes in group $k$.

Lines 16-22 are only executed by processes in group $k$. They obtain the set of values sent by all processes at line 15 and do classification based the size of this set. To prevent a Byzantine process from sending arbitrary values at line 15, each process in group $k$ only accepts the set $R_j$ from process $j$ if $R_j$ is a subset of $U_{i,1}^k$. If process $j$ is correct, this condition must hold by property 2 of gradecast. So the sets sent from correct processes will always be accepted. Then, the set $T$ at line 18 must contain all the sets sent from correct processes. Since each such set is the safe value set of a process for the slave group, each process in group $k$ is actually doing classification based on the safe sets of processes for the slave group. Lines 20-21 is the classification step. If the size of $T$ is greater than $k$, then the process is classified as a master and updates its value to be the union of its current value and the set of values gradecast by processes in group $k$ with score at least 1. Otherwise, its value is updated to be the set of values gradecast by processes in group $k$ with score 2. Its label is updated based on whether the process is master or slave. 

\begin{table*}[h] 
\caption{Notations}
\begin{tabular}{ |c|c|} 
 \hline
 \textbf{Variable} & \textbf{Definition} \\ \hline
 $V_i^r$ &The value set held by process $i$ at the beginning of round $r$ \\ \hline 
$F_i^r$ & \makecell{The safe value map of $i$ at the beginning of round $r$ \\
$F_i^r[k]$ is $i$'s safe value set for group $k$} \\ \hline
$S_i^r$ & \makecell{Auxiliary array of size $n$ derived from $F_i^r$. \\ $S_i^r[j]$ is process $i$'s safe value set for process $j$ at the beginning of round $r$ \\ i.e., $S_i^r[j] = F_i^r[k_j]$ where $k_j$ is the label of process $j$ at round $r$} \\ \hline 
$SF_{k}^r$ & \makecell{Auxiliary variable denoting the union of the safe value set of each correct process \\ for group $k$  at the beginning of round $r$, i.e., $SF_k^r := \{F_i^r[k] ~| ~ i \in C\}$}\\ \hline
$T_i^r$ & The value of variable $T$ in process $i$ at line 18 of round $r$ \\ \hline 
\end{tabular}
\label{tab:variable_logf}
\end{table*}

For any group $G$ with label $k$, let $S(G)$ and $M(G)$ denote the slave group and the master group. Let $s(k, r)$ denote the label of $S(G)$ and $m(k, r)$ denote the label of $M(G)$. The variables we use for the proof is given in Table \ref{tab:variable_logf}. The classifier procedure has the following properties. 

\begin{lemma} \label{lem:cls_logf}
Let $G$ be a group at round $r$ with label $k$ and $G \cap C \not = \emptyset$. Let $L$ and $R$ be two nonnegative integers such that $L \leq k \leq R$. If $L < |V_i^{r}| \leq R$ for each correct process $i \in G$, and $|SF_{k}^r| \leq R$, then \\
(p1) For each correct process $i \in M(G)$, $k < |V_i^{r + 1}| \leq R$\\
(p2) For each correct process $i \in S(G)$, $L < |V_i^{r + 1}| \leq k$\\
(p3) $SF_{m(k,r)}^{r + 1} \subseteq SF_{k}^r$ \\
(p4) $SF_{s(k,r)}^{r + 1} \subseteq SF_{k}^r$ \\
(p5) $|SF_{m(k,r)}^{r + 1}| \leq R$ \\ 
(p6) $|SF_{s(k,r)}^{r + 1}| \leq k$ if $S(G) \cap C \not = \emptyset$\\
(p7) For each correct process $j \in M(G)$, $SF_{s(k,r)}^{r + 1} \subseteq V_j^{r + 1}$\\
(p8) For each correct process $i \in G$, $V_i^{r + 1} \subseteq SF_k^r$ 
\begin{proof}
{\bf (p1)-(p2)}: Immediate from the {\em Classifier} procedure. 

{\bf (p3)}: Consider line 14 of round $r$. Each correct process $p$ updates its safe value set for group $m(k, r)$ as $S_p[m(k, r)] := S_p[k] \cup U_{p,1}^k$, with $U_{p,1}^k$ being the values gradecast by processes in group $k$ with score at least 1. Thus, any value $v \in U_{p,1}^k$ must be in the safe set for group $k$ of at least one correct process by property 3 of the {\bf SetGradecast} procedure. Then, we have $U_{p,1}^k \subseteq SF_k^r$. Hence, $SF_{m(k, r)}^{r + 1} \subseteq SF_{k}^r$ 

{\bf (p4)}: Consider an arbitrary value $v \in SF_{s(k,r)}^{r + 1}$. From line 14 of round $r$, we know that $v$ must be gradecast by some process in group $k$ with score 2 at round $r$, which implies that value $v$ must be in the safe value set for group $k$ of at least one correct process. Then, value $v$ must be in $SF_k^{r}$. Thus, $SF_{s(k,r)}^{r + 1} \subseteq SF_k^{r}$.

{\bf (p5)}: Immediately implied by $(p3)$ and the condition that $|SF_{k}^r| \leq R$. 

{\bf (p6)}: Consider an arbitrary correct processes $i \in S(G)$ at round $r$. We first show that $S_p[s(k,r)] \subseteq T_i^r$ for each correct process $p \in [n]$. At line 14, $S_p[s(k,r)]$ is updated to be $U_{p,2}^k$, the set of values gradecast by processes in group $k$ with score 2. At line 15, process $p$ send its $U_{p,2}^k$ to process $i$. At line 17, process $i$ must receive $R_p$ from process $p$ and $R_p$ must be subset of $U_{i,1}^k$ of process $i$ since $p$ is correct. Thus, $S_p[s(k,r)] \subseteq T_i^r$. Therefore, we have $SF_{s(k,r)}^{r + 1} \subseteq T_i^r$. Since $i \in S(G)$, we have $|T_i^r| \leq k$, which implies that $|SF_{s(k,r)}^{r + 1}| \leq k$. 

{\bf (p7)}: Consider an arbitrary value $v \in SF_{s(k,r)}^{r + 1}$, from line 14 of the algorithm, $v$ must be gradecast by some process in $G$ at round $r$ and assigned score of 2 by at least one correct process. By property 2 of the {\bf SetGradecast} algorithm, the score of $v$ assigned by any correct process must be at least 1. Specifically, the score of $v$ assigned by process $j$ must be at least 1, which means $v$ must be included into $V_j^r$ by process $j$ at line 20. 

{\bf (p8)}: $V_i^{r + 1}$ is either the set of values gradecast by processes in group $G$ with score 2 or the set of values with score at least 1. In either case, by property of gradecast, each value in $V_i^{r + 1}$ must be contained in $SF_k^{r}$. 
\end{proof}
\end{lemma}

The following lemma guarantees that if a value in the value set of a correct process is contained in the safe value set of each correct process, then this value remains in the value set of the correct process. With the following lemma and property $(p4)$ and $(p7)$ of Lemma \ref{lem:cls_logf}, we have $SF_{s(k,r)}^{t} \subseteq V_j^t$ for any round $t \geq r$. 
\begin{lemma} \label{lem:correct_value_log_f}
Consider an arbitrary value $v \in V^r_j$ of correct process $j$ at round $r$. If $v \in S_i^r[j]$ for each correct process $i$, then we have \\
(1) $v \in S_i^t[j]$ for each correct process $i$ and any round $t \geq r$. \\
(2) $v \in V_j^t$ for any round $t \geq r$. 
\begin{proof}
 By induction on the round number. For the base case $t = r$, the claim holds. Consider round $t > r$ and assume that (1) and (2) are satisfied for any round before $t$. At round $t$, since value $v \in S_i^{t}$ for each correct process $i$ and $v \in V_j^t$ by induction hypothesis, then value $v$ must be gradecast by process $j$ and assigned score of 2 by each correct process $i$ at round $t$, by property 1 of {\bf SetGradecast}. Thus, $v$ must be included into $S_i^{t + 1}[j]$ by each correct $i$ at line 14 and $v$ must included into $V^{t + 1}_j$ at line 20 or 21 by process $j$.
\end{proof}
\end{lemma}

\begin{lemma}\label{lem:dec}
Let $G$ be a group of processes at round $r$ with label $k$ and $G \cap C \not = \emptyset$. Then \\
(1) for each correct process $i \in G$, $k - \frac{f}{2^r} < |V_i^{r}| \leq k + \frac{f}{2^r}$ \\
(2) $|SF_k^r| \leq k + \frac{f}{2^r}$

\begin{proof}
By induction on $r$. Consider the base case with $r = 1$, $k = k_0 = n - \frac{f}{2}$.  After the initial round, each correct process $i$ must have at least $n - f$ values in its value set $V_i^1$. By the property of gradecast, each process (could be Byzantine) can only introduce one value into the value sets of all correct processes at the initial round. Thus, $|V_i^1| \leq n$. So, (1) is proved. For (2), similarly, each process (could be Byzantine) can only introduce one value into $SF_{k_0}^r$. Thus, $|SF_{k_0}^1| \leq n = k_0 + \frac{f}{2}$. 

For the induction step, assume the above lemma holds for all groups with at least one correct process at round $r - 1$. Consider an arbitrary group $G$ at round $r > 1$ with label $k$. Let $G'$ be the parent group of $G$ at round $r - 1$ with label $k'$. Consider the \textit{Classifier} procedure executed by all processes in $G'$ with label $k'$. By induction hypothesis, we have: 

(1) for each correct process $i \in G'$, $k' - \frac{f}{2^{r - 1}} < h(V_i^{r - 1}) \leq k' + \frac{f}{2^{r - 1}}$ 

(2) $|SF_{k'}^{r - 1}| \leq k' + \frac{f}{2^{r - 1}}$.

 Let $L = k' - \frac{f}{2^{r - 1}}$ and $R = k' + \frac{f}{2^{r - 1}}$, then (1) and (2) are exactly the conditions of Lemma \ref{lem:cls_logf}. Consider the following two cases: 

 Case 1: $G = M(G')$. Then $k = m(k', r - 1) = k' + \frac{f}{2^r}$. From ($p1$) and ($p5$) of Lemma \ref{lem:cls_logf}, we have:
 
 (1) for each correct process $i \in G$, $k - \frac{f}{2^r} < h(V_i^r) \leq k + \frac{f}{2^r}$ 
 
 (2) $|SF_{k}^r| \leq k + \frac{f}{2^r}$. 

 Case 2: $G = S(G')$. Then $k = s(k', r-1) = k' - \frac{f}{2^r}$. From ($p2$) and ($p6$) of Lemma \ref{lem:cls_logf}, we have the same equations. 
 \end{proof}
\end{lemma}

\begin{lemma}\label{lem:same_group}
Let $i$ and $j$ be two correct processes that are within the same group $G$ at the beginning of round  $\log f + 1$. Then $V_i^{\log f + 1} = V_j^{\log f + 1}$ .
\begin{proof}
Let $G'$ be the parent of $G$ with label $k'$. Assume without loss of generality that $G = M(G')$. The proof for the case $G = S(G')$ follows in the same manner. Since $G'$ is a group at round $\log f$, by Lemma \ref{lem:dec}, we have: \\
 (1) for each correct process $p \in G'$, $k' - 1 < |V_p^{\log f}| \leq k' + 1$, and\\
 (2) $|SF_{k'}^{\log f}| \leq k' + 1$

 Since $i \in G'$ and $j \in G'$, (1) holds for both process $i$ and $j$. By the assumption that $G = M(G')$, at round $\log f$,  process $i$ and $j$ execute the \textit{Classifier} procedure with label $k'$ in group $G'$ and are classified as \textit{master}. Let $L = k' - 1$ and $R = k' + 1$, then by applying Lemma \ref{lem:cls_logf}($p1$) we have $k' < |V_i^{\log f + 1}| \leq k' + 1$ and $k' < |V_j^{\log f + 1}| \leq k' + 1$, thus $|V_i^{\log f + 1}| = |V_j^{\log f + 1}| = k' + 1$. Similarly, by ($p8$) of Lemma \ref{lem:cls_logf}, we have $V_i^{\log f + 1} \subseteq SF_{k'}^{\log f}$ nd $V_j^{\log f + 1} \subseteq SF_{k'}^{\log f}$. Hence, $|V_i^{\log f + 1} \cup V_j^{\log f + 1}| \leq k' + 1$. Thus, $V_i^{\log f + 1} = V_j^{\log f + 1}$. 
\end{proof}
\end{lemma}

\begin{lemma} \label{lem:comp_log_f}
(Comparability) For any two correct $i$ and $j$, we have either $y_i \leq y_j$ or $y_j \leq y_i$. 
\begin{proof}
We show either $V_i^{\log f + 1} \subseteq V_j^{\log f + 1}$ or $V_j^{\log f + 1} \subseteq V_i^{\log f + 1}$. Let $G$ denote the last group with label $k$ that both $i$ and $j$ belongs to. Let $r$ be the last round that both $i$ and $j$ belong to group $G$. Without loss of generality, suppose $i \in S(G)$ and $j \in M(G)$. We have $V_i^{\log f + 1} \subseteq SF_i^{\log f} \subseteq SF_{S(k)}^r$. Consider round $r$ of the algorithm. Consider an arbitrary value $v \in SF_{S(G)}^r$. The value $v$ must be gradecast by some process in $S(G)$ and assigned score 2 by some correct process by line 3 of the classifier. That is, there exists a correct process $t$ with $c_v^t = 2$. Then, we have $c_v^k \geq 1$ for any correct process $k$ by property 2 of the {\bf SetGradecast} algorithm. Hence, $v \in S^r_k[j]$ for any correct process $k$ by line 22 of the algorithm. Also, we have $v \in V^r_j$ by line 24. By Lemma \ref{lem:correct_value_log_f}, we can derive that $v \in V_j^{\log f + 1}$. Since $v$ is an arbitrary value from $SF_{S(G)}^r$, we must have $SF_{S(G)}^r \subseteq V_j^{\log f + 1}$. Thus, $V_i^{\log f + 1} \subseteq V_j^{\log f+ 1}$. The case when $j \in S(G)$ and $i \in M(G)$, we can similar obtain $V_j^{\log f + 1} \subseteq V_i^{\log f + 1}$. 
\end{proof}
\end{lemma}

\logfTheorem*

\begin{proof}
{\bf Downward-Validity}. After the initial round, the input $x_i$ of process $i$ must be in the safe value set of each correct process for process $i$. Then, by Lemma \ref{lem:correct_value_log_f}, $x_i \in V_i^{\log f + 1}$. Thus, $x_i \leq y_i$. 

{\bf Comparability} follows from Lemma \ref{lem:comp_log_f}. 

{\bf Upward-Validity}. After the initial round, we have $SF_{k_0}^1 \subseteq \{x_i ~|~ i \in C\} \cup B$, where $B \subset X$ and $|B| \leq t$, since each process can introduce at most one value in $SF_{k_0}^1$ by property 1 of gradecast. Since $V_j^{\log f + 1} \subseteq SF_{k_0}^1$ for each $j \in C$, we have $\sqcup \{y_i ~ | ~ i \in C\} \leq \sqcup(\{x_i ~ | ~ i \in C\} \cup B)$, where $B \subset X $ and $|B| \leq t$.  \\
\end{proof}

\begin{corollary}
There is a $5\log f + 4$ rounds algorithm for the authenticated BLA problem in synchronous systems which can tolerate $f < \frac{n}{2}$ Byzantine failures, where $n$ is the number of processes in the system. The algorithm takes $O(n^2 \log n)$ messages. 
\end{corollary}

\section{Conclusion} 
We have presented three algorithms for the Byzantine lattice agreement problem in synchronous systems. The first algorithm takes at most $\min\{3h(X) + 6, 6\sqrt{f} + 6\}$ rounds and has early stopping property. The second algorithm takes $3\log n + 3$ rounds. The third algorithm takes $4 \log f + 3$ rounds. The latter two algorithms do not have early stopping property. The $O(\log f)$ rounds upper bound matches what we have for the crash failure setting. For future work, the following questions are interesting: 1) Can we improve the upper bound or prove some lower bound on the round complexity? There does not exist any lower bound for the lattice agreement problem at all, even in the crash failure model. 2) Can we improve the message complexity? 3) Is there an algorithm for the authenticated Byzantine lattice agreement problem which can tolerate any number of failures and takes less than $f + 1$ rounds? It is known to be impossible for the authenticated Byzantine agreement problem~\cite{dolev1983authenticated}.

\bibliography{ref}

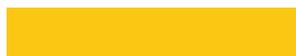
\end{document}